\documentclass[twocolumn,conference,letterpaper]{IEEEtran}
\usepackage[left=0.75in,right=0.75in,top=1in,bottom=0.75in]{geometry}

\usepackage{cite}
\usepackage{amssymb,amsmath,latexsym,amsfonts,amsthm,mathtools}

\usepackage{graphicx}
\usepackage{float,subcaption}
\usepackage{textcomp}
\usepackage{xcolor}
\usepackage{hyperref}
\hypersetup{colorlinks, breaklinks, citecolor=blue, linkcolor=blue, urlcolor=blue}
\usepackage[nameinlink]{cleveref}
\Crefformat{figure}{#2Fig.~#1#3}
\Crefmultiformat{figure}{Figs.~#2#1#3}{ and~#2#1#3}{, #2#1#3}{ and~#2#1#3}

\theoremstyle{plain}
\newtheorem{theorem}{Theorem}
\newtheorem{lemma}{Lemma}
\newtheorem{corollary}{Corollary}

\theoremstyle{remark}
\newtheorem{remark}{Remark}

\theoremstyle{definition}
\newtheorem{problem}{Problem}
\newtheorem{definition}{Definition}

\usepackage{mathrsfs}
\usepackage{newtxmath}
\usepackage{booktabs}

\newcommand{\R}{\mathbb{R}}
\newcommand{\DMC}{\mathrm{DMC}}
\newcommand{\eqdef}{\triangleq}

\IEEEoverridecommandlockouts

\begin{document}

\title{Risk-Aware Input-Constrained Safe Intercept Guidance
Against Multiple Moving Defenders}

\author{Praveen Kumar Ranjan,~\IEEEmembership{Member,~IEEE},
        Abhinav Sinha,~\IEEEmembership{Senior Member,~IEEE},\\
        Yongcan Cao,~\IEEEmembership{Senior Member,~IEEE},
        Alexander Von Moll,~\IEEEmembership{Senior Member,~IEEE},\\
        and Isaac E. Weintraub,~\IEEEmembership{Senior Member,~IEEE}
        \thanks{P. K. Ranjan and Y. Cao are with the Unmanned Systems Lab,
        Department of Electrical and Computer Engineering, The University of
        Texas at San Antonio, San Antonio, TX 78249, USA (e-mails:
        \url{praveen.ranjan@my.utsa.edu}, \url{yongcan.cao@utsa.edu}). A. Sinha
        is with the Guidance, Autonomy, Learning, and Control for Intelligent
        Systems (GALACxIS) Lab, Department of Aerospace Engineering, University of Cincinnati, OH 45221, USA (e-mail:
        \url{abhinav.sinha@uc.edu}). A. Von Moll and I. E. Weintraub are with
        the Mission Autonomy Branch, Air Force Research Laboratory,
        Wright-Patterson AFB, OH 45433, USA (e-mails:
        \url{alexander.von_moll@us.af.mil}, \url{isaac.weintraub.1@us.af.mil}).}
        }

\maketitle
% \thispagestyle{empty}

%------------------------------- abstract -------------------------------------
\begin{abstract}
This paper develops a risk-aware guidance law for an attacker to
intercept a stationary target in the presence of multiple moving defenders while respecting the control input limits. We represent defender threats via engagement zones (EZs) (regions of state-space from which attacker interception is feasible) and employ the dynamic maneuvering cue (DMC) (minimum heading correction to exit EZs) as a deterministic measure of geometric risk. Safety is enforced through a prescribed bound on the DMC to enable controlled EZ penetration subject to a specified risk tolerance. For multiple defenders, the corresponding safety constraints are combined through a smooth-minimum aggregation. The proposed guidance strategy drives the attacker toward the target when the engagement remains safe and transitions to an avoidance maneuver as the prescribed safety boundary is approached. The resulting online optimization-free guidance law directly bounds the applied lateral acceleration and initiates a bounded preemptive maneuver before corrective action is required. Closed-loop analysis establishes preservation of the prescribed risk constraints under pointwise bounded-input feasibility and finite-time target interception under a sustained target-accessibility condition. Numerical studies validate the proposed strategy and illustrate the tradeoff between allowable geometric risk and interception performance.
\end{abstract}


%==============================================================================
\section{Introduction}

Autonomous systems are increasingly being deployed in contested environments and must accomplish mission objectives while remaining resilient to adversarial threats \cite{10660561}. In missions such as target interception, air defense, and precision strike, guidance design must therefore account not only for terminal accuracy but also for operational constraints imposed by the defended environment. Target-Attacker-Defender (TAD) games have provided a useful framework for studying such conflicts \cite{isaacs1999differential,ratnoo2011lineofsight,
shima2011optimal}. Existing attacker-side TAD and pursuit-evasion studies have developed strategies using differential games, optimal guidance, robust control, and data-driven methods to enable target interception while avoiding defenders \cite{doi:10.2514/1.61832,doi:10.2514/1.51611, doi:10.2514/6.2025-1902,sinha2022three,11320436}. However, such approaches often rely on assumptions about defender strategies, simplified engagement dynamics, or tractable game structures that may become restrictive in adversarial settings. 

Alternatively, interception and obstacle-avoidance methods often enforce safety through minimum-distance or fixed stand-off constraints \cite{doi:10.2514/1.G003223,doi:10.2514/6.2026-0121,doi:10.2514/1.37030,doi:10.2514/1.G003157}. However, such distance-based safety constraints do not fully capture the directional nature of defender threats and can lead to conservative avoidance behavior. Geometric approaches, such as velocity-obstacle \cite{doi:10.1177/027836499801700706} and collision-cone methods \cite{709600}, use relative motion to identify velocity or geometric regions associated with future collision risk. However, these formulations
are primarily intended for collision or kinematic conflict avoidance and do not directly represent the finite interception capability of a defender. For interception scenarios involving active defenders, engagement zones offer a geometry-based threat representation that characterizes when the attacker becomes vulnerable to interception \cite{10365808,doi:10.2514/1.I011394,doi:10.2514/1.I011593}. Therefore, the EZ boundary differentiates threat levels that may arise at the same attacker-defender separation and result in less conservative safety constraints.

Although EZs reduce the conservatism of fixed stand-off constraints, the resulting unsafe region can still be conservative because penetrations with different levels of recoverability are treated equivalently. This motivates the use of the dynamic maneuvering cue (DMC) \cite{vonmoll2025dmc}, which measures recoverability in terms of the heading correction needed to restore a safe engagement condition. This enables the allowable threat exposure to be specified through a prescribed DMC tolerance, rather than by strict EZ avoidance. Extending this risk-aware formulation to practical engagements further requires that the prescribed safety maneuvers remain feasible under actuator constraints. Existing approaches account for actuator limitations either indirectly through quadratic costs or feasibility conditions \cite{RUSNAK20119349,9682938}, or by explicitly incorporating bounded acceleration or actuator saturation into the vehicle dynamics, e.g., \cite{9682938}, or through an admissibility preserving input realization (APIR) \cite{Kumar2026APC,SinhaKumar2026AsymmetricConsensus,Sinha2026DirectedSafeCoordination}. However, simultaneously enforcing DMC-based risk constraints for multiple moving defenders under finite control authority introduces additional challenges.  Therefore, this work is motivated by the need for a guidance strategy that simultaneously quantifies geometric risk, accommodates controlled EZ penetration,
handles multiple moving defenders, and respects the attacker's physical maneuvering limits without relying on defender control policies or online numerical optimization. The main contributions of this work are summarized as follows:
\begin{itemize}
    \item We develop a risk-aware intercept guidance framework that bounds defender exposure through a prescribed DMC tolerance to enables controlled penetration inside the threat zone rather than strict avoidance.

    \item We propose an online optimization-free safety-critical guidance law for multiple moving defenders that aggregates the individual risk constraints through a smooth minimum and combines a preemptive avoidance maneuver with the target-interception maneuver, while directly respecting the attacker’s lateral-acceleration bound.

    \item We establish risk-constraint and input-bound preservation under pointwise feasibility, together with finite-time target interception once sustained target-directed motion is permitted by the risk constraints.

\end{itemize}
\section{Problem Formulation}
Consider a planar engagement consisting of a mobile attacker $A$, $n$ moving defender platforms $D_i$, indexed by $\mathcal{D}=\{1,\dots,n\}$, and a stationary target $T$, as shown in Fig.~\ref{fig:enggeom}.  Each defender is modeled as a platform moving at a constant speed $v_i\ge0$ and possessing a defensive engagement capability characterized by speed $v_i^I>v_A$, range $R_i\in\R_{\ge0}$, and capture radius $c_i\in\R_{\ge0}$. Accordingly, each defender can affect the attacker within a maximum engagement distance of $R_i+c_i$ from its current position, depicted by the dotted circle around each defender in Fig.~\ref{fig:enggeom}. 
\begin{figure}[h!]
    \centering
    \includegraphics[width=.5\linewidth]{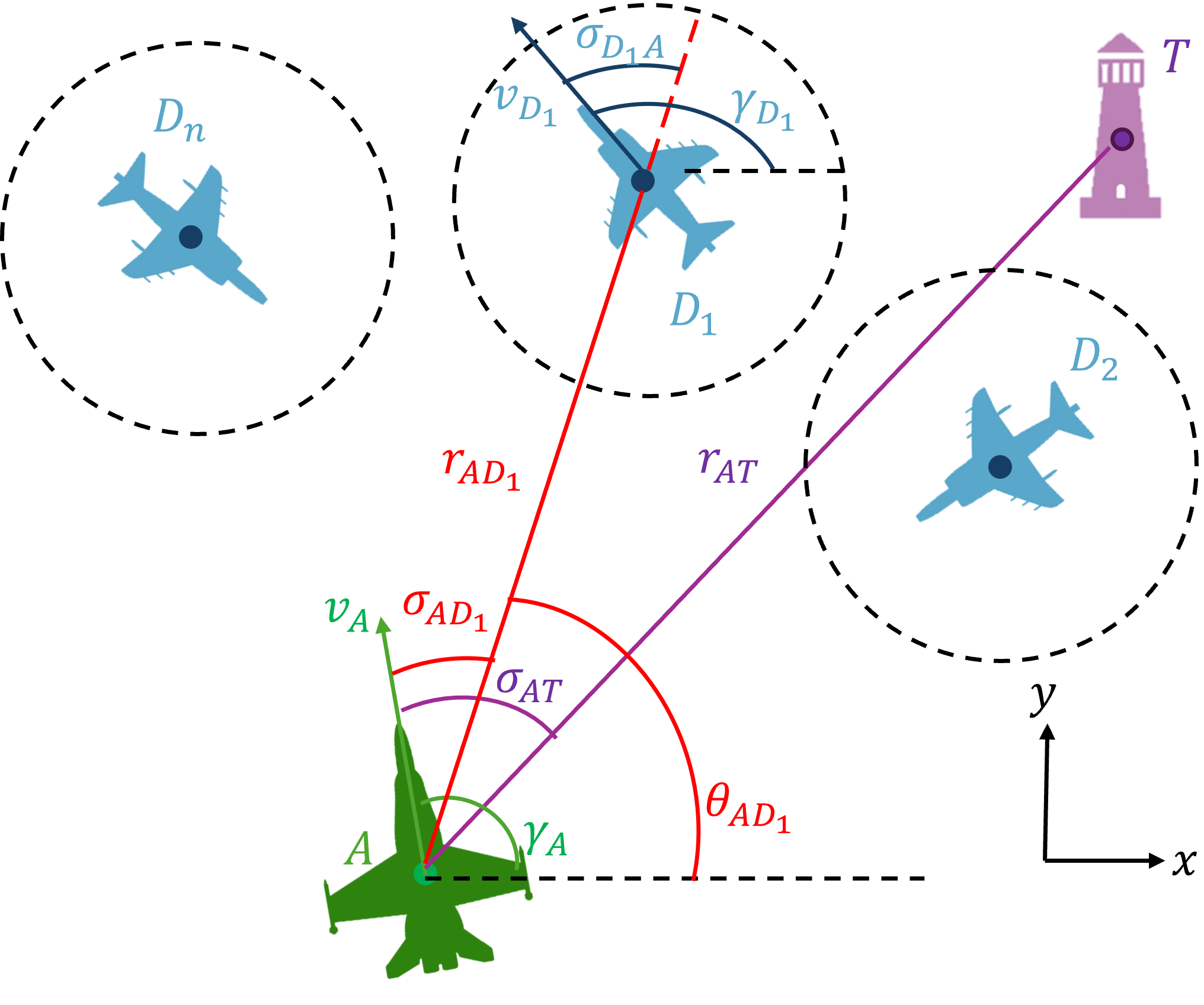}
    \caption{Attacker-Target-Defender engagement geometry.}
    \label{fig:enggeom}
\end{figure}

To describe the engagement geometry in relative coordinate frame, let attacker to the j\textsuperscript{th} entity relative distance and line-of-sight (LOS) angle be denoted by $r_{Aj}>0$ and $\theta_{Aj}$, respectively, where $j\in\mathcal{D}\cup\{T\}$. Denoting $\gamma_i \in (-\pi,\pi]$ as the heading angle of the moving entities, we define the corresponding lead as the angle between the LOS and the velocity of the vehicle given by,  $\sigma_{Aj} = \gamma_A-\theta_{Aj},
\; j\in\mathcal{D}\cup\{T\}$ and $\sigma_{iA} = \gamma_i-\theta_{Ai}, \;i\in\mathcal{D}.$
% \begin{align}
% \sigma_{Aj} &= \gamma_A-\theta_{Aj},
% \qquad j\in\mathcal{D}\cup\{T\}, \label{eqn:lead_attacker}\\
% \sigma_{iA} &= \gamma_i-\theta_{Ai},
% \qquad i\in\mathcal{D}. \label{eqn:lead_defender}
% \end{align}
It is important to note that in the above equation the defenders lead angle is defined with respect to the same attacker-to-defender LOS $\theta_{Ai}$, such that a defender moving directly toward the attacker satisfies $\sigma_{iA}=\pi$. 

Assuming the attacker to be a constant speed non-holonomic vehicle, the relative motion between the attacker and the i\textsuperscript{th} defender or target is governed by,
\begin{align}
\dot r_{Ai}
&=v_i\cos\sigma_{iA}-v_A\cos\sigma_{Ai},
\qquad ,  \label{eqn:rel_dyn_1}\\
\dot\theta_{Ai}
&=\frac{v_i\sin\sigma_{iA}-v_A\sin\sigma_{Ai}}
{r_{Ai}}, \label{eqn:rel_dyn_2}\\
\dot\sigma_{Ai}
&=\frac{a_A}{v_A}-\dot\theta_{Ai},
\qquad i\in\mathcal D\cup\{T\}, \label{eqn:rel_dyn_3}
\end{align}
where $a_A$ denotes the attacker's lateral acceleration that is the control input to the vehicle model. For the stationary target, $v_T=0$ and therefore, target heading or corresponding $\sigma_{TA}$ is not required.

The engagement capability for each defender is evaluated by considering an interceptor launched from the defender's instantaneous position. Once launched, the interceptor moves independently of the carrying platform with constant speed $v_i^I$, maximum engagement range $R_i$, and capture radius $c_i$. Thus, the motion of the defender platform changes the instantaneous launch location of the interceptor, but does not affect the motion of an interceptor after launch. This interception geometry induces a heading-dependent region of vulnerability for the attacker, referred to as the engagement zone.

\begin{definition}[Engagement Zone]
The engagement zone (EZ) of the i\textsuperscript{th} defender is defined as the set of attacker-relative states from which the attacker interception is feasible by defender, if the attacker maintains its current heading. \label{def:ez}
\end{definition}
To characterize this region, we define the attacker-to-interceptor speed ratio as,$ \mu_i\triangleq \frac{v_A}{v_i^I}\in(0,1).$ For a range-limited defender, the boundary of the corresponding EZ is given in closed form as \cite{doi:10.2514/1.I011394},
\begin{equation}
\rho_i(\sigma_{Ai})=\mu_iR_i\!\left[\cos\sigma_{Ai}
+\sqrt{\cos^2\sigma_{Ai}-1+\frac{(R_i+c_i)^2}{\mu_i^2R_i^2}}\;\right],
\label{eq:rho}
\end{equation}
Crossing this boundary places the attacker within the defender’s feasible capture region. Therefore, based on \Cref{def:ez}, the EZ of a defender $D_i$ with respect to the attacker is defined as, $\mathcal{E}_i := \left\{ (r_{Ai}, \sigma_{Ai}) \in \mathbb{R}_{\geq 0} \times (-\pi, \pi] \,\middle|\, r_{Ai} \leq \rho_i(\sigma_{Ai}) \right\}.$ and the corresponding EZ margin is defined as, $b_i=r_{Ai}-\rho_i\left(\sigma_{Ai}\right)$. Therefore, $b_i>0$, $b_i=0$, and $b_i<0$ correspond to states outside, on, and inside the EZ, respectively. The EZ boundary determines whether the attacker's current heading is admissible, but does not quantify the maneuver required to recover an admissible heading when the attacker lies inside the EZ. To capture this additional geometric information, we employ the dynamic maneuvering cue (DMC) developed in \cite{vonmoll2025dmc}.

For a fixed attacker-defender range $r_{Ai}$, the DMC represents the minimum instantaneous heading correction required to reach the nearest EZ boundary. To characterize this correction, we first identify the heading at which the attacker lies exactly on the EZ boundary. This lead angle is given as, 
\begin{align}
\sigma_i^\ast\left(r_{Ai}\right)
&=\cos^{-1}\left(\frac{r_{Ai}^2+\mu_i^2R_i^2-(R_i+c_i)^2}
{2\mu_iR_i r_{Ai}}\right)
,
\end{align}
obtained by solving $ \rho_i(\sigma_i^\ast)=r_{Ai} $. The above boundary lead angle exists only within the annular region $\underline r_i\le r_{Ai}\le\overline r_i,$ where$ \underline r_i=(1-\mu_i)R_i+c_i,\qquad \overline r_i=(1+\mu_i)R_i+c_i.$ Within this region, the EZ boundary satisfies
$|\sigma_{Ai}|=\sigma_i^\ast(r_{Ai}),$
and the minimum heading correction required to reach the nearest admissible heading is therefore \(\sigma_i^\ast(r_{Ai})-|\sigma_{Ai}|\) whenever the attacker lies inside the EZ. 
\begin{definition}[Dynamic Maneuvering Cue]
\label{def:dmc}
For defender $D_i$, the dynamic maneuvering cue (DMC) is the minimum instantaneous heading correction required for the attacker to reach the nearest EZ-admissible heading. It is given as \cite{vonmoll2025dmc},
\begin{equation}
|\DMC_i|\eqdef
\begin{cases}
+\infty, & r_{Ai}<\underline r_i,\\
\max\{0,\sigma_i^\ast(r_{Ai})-|\sigma_{Ai}|\},
 & \underline r_i\le r_{Ai}\le\overline r_i,\\
0, & r_{Ai}>\overline r_i.
\end{cases}
\label{eq:dmc}
\end{equation}
\end{definition}
Thus, $|\mathrm{DMC}_i|=0$ indicates that the current attacker heading is already EZ-admissible, while a finite positive value quantifies the minimum heading correction required to reach the nearest EZ boundary. For $r_{Ai}<\underline r_i$, no admissible heading exists and the DMC is therefore taken as unbounded.

We interpret $|\DMC_i|$ as a deterministic measure of geometric risk, since it quantifies the minimum instantaneous heading correction required to recover EZ admissibility. Accordingly, for each defender $D_i$, we prescribe a risk tolerance $\varepsilon_i\in[0,\pi)$ and impose
\begin{align}
|\DMC_i(t)|\leq \varepsilon_i.
\label{eq:dmc_constraint}
\end{align}
This constraint permits controlled penetration into an EZ provided that the heading correction required to recover an admissible state remains within the prescribed risk tolerance. The choice $\varepsilon_i=0$ recovers strict EZ avoidance, whereas $\varepsilon_i>0$ allows increasing geometric exposure while bounding the required recovery maneuver.

Although \eqref{eq:dmc_constraint} provides a direct bound on the required heading correction, its piecewise angular form is not convenient for subsequent control design. We therefore seek an equivalent representation in terms of the EZ boundary. For each defender, define
\begin{align}
\chi_i(\sigma_{Ai};\varepsilon_i)
&\triangleq
\min\left\{|\sigma_{Ai}|+\varepsilon_i,\pi\right\},
\end{align}
and the corresponding radial margin
\begin{align}
B_i(r_{Ai},\sigma_{Ai};\varepsilon_i)
&\triangleq
r_{Ai}
-\rho_i\left(
\chi_i(\sigma_{Ai};\varepsilon_i)
\right).
\label{eq:Bi}
\end{align}
The quantity $\chi_i$ represents the attacker lead angle augmented by the prescribed risk tolerance, while $B_i$ measures the radial clearance from the corresponding EZ boundary.

\begin{lemma}
\label{lem:dmc_radial}
Enforcing the non-negative radial margin in \eqref{eq:Bi} is equivalent to enforcing the DMC constraint in \eqref{eq:dmc_constraint}.
\end{lemma}
\begin{proof}
From \eqref{eq:rho}, $\rho_i(\cdot)$ is strictly decreasing on
$[0,\pi]$, with
$\rho_i(0)=\overline r_i$ and
$\rho_i(\pi)=\underline r_i$. For $\underline r_i\leq r_{Ai}\leq\overline r_i$, the DMC definition
in \eqref{eq:dmc} gives $
\sigma_i^\ast(r_{Ai})
\leq |\sigma_{Ai}|+\varepsilon_i $.
Since $\sigma_i^\ast(r_{Ai})\leq\pi$, this is equivalently, $\sigma_i^\ast(r_{Ai})
\leq \chi_i(\sigma_{Ai};\varepsilon_i).$
Using the monotonicity of $\rho_i(\cdot)$ and $\rho_i(\sigma_i^\ast)=r_{Ai}$, the above condition is equivalent to $
r_{Ai}\geq
\rho_i\!\left(\chi_i(\sigma_{Ai};\varepsilon_i)\right),$ which, from \eqref{eq:Bi}, is precisely $B_i\geq0$. For $r_{Ai}>\overline r_i$, \eqref{eq:dmc} gives
$|\DMC_i|=0$, while
$\rho_i(\chi_i)\leq\rho_i(0)=\overline r_i<r_{Ai}$, and hence
$B_i>0$. Conversely, for $r_{Ai}<\underline r_i$,
$|\DMC_i|=+\infty$, while
$\rho_i(\chi_i)\geq\rho_i(\pi)=\underline r_i>r_{Ai}$, yielding
$B_i<0$. Therefore, $B_i\geq0$ is equivalent to
$|\DMC_i|\leq\varepsilon_i$ in all three range regimes. This concludes the proof.
\end{proof}

Lemma~\ref{lem:dmc_radial} allows the DMC constraint to be represented by the radial condition $B_i\geq0$, which is more convenient for subsequent guidance design. In particular, when $\varepsilon_i=0$, $ B_i(r_{Ai},\sigma_{Ai};0) = r_{Ai}-\rho_i(\sigma_{Ai}) = b_i, $ and the original EZ-avoidance condition is recovered. For \(\varepsilon_i>0\), the admissible region is enlarged to include states inside the EZ for which the required recovery turn remains within the prescribed DMC budget.

\begin{problem}
\label{prob:guidance}
For prescribed risk tolerances $\varepsilon_i\in[0,\pi)$, design a guidance law for the attacker to ensure that the attacker reaches a prescribed target capture radius $r_c>0$ in finite time, i.e., there exists $t_f<\infty$ such that $r_{AT}(t_f)\leq r_c$. During the engagement, the attacker must maintain safety by  ensuring DMC constraint with respect to every defender, that is, $B_i(r_{Ai}(t),\sigma_{Ai}(t);\varepsilon_i)\geq0$, $\forall i\in\mathcal{D}$ and $t\in[0,t_f]$. The attacker must also satisfy the lateral-acceleration constraint $|a_A(t)|<a_{\max}$ for all $t\in[0,t_f]$. 
\end{problem}

\section{Risk-aware Guidance Design}
This section develops the proposed multi-defender risk-aware guidance law and incorporates the acceleration constraint into the control design to establish the resulting safety and interception properties. The dependence of $B_i$ on $|\sigma_{Ai}|$ renders the margin nonsmooth at the head-on configuration. Therefore, we select a fixed direction $s_i\in\{-1,1\}$ at the beginning of the engagement and retain it throughout the engagement. This yields the following smooth surrogate of the original risk margin,
\begin{align}
\widehat B_i = r_{Ai}-\rho_i\!\left(\varepsilon_i+s_i\sigma_{Ai}\right).
\label{eq:Bi_hat}
\end{align}

\begin{lemma}
\label{lem:conservative_branch}
For every fixed $s_i\in\{-1,1\}$, the surrogate margin in \eqref{eq:Bi_hat} satisfies $\widehat B_i\leq B_i$. 
\end{lemma}
\begin{proof}
For either $s_i\in{-1,1}$, $|\varepsilon_i+s_i\sigma_{Ai}|
\leq \varepsilon_i+|\sigma_{Ai}|$. If $\varepsilon_i+|\sigma_{Ai}|\leq\pi$, then $|\varepsilon_i+s_i\sigma_{Ai}|\leq\chi_i$. Evenness and monotonic decrease of $\rho_i$ on $[0,\pi]$ give $\rho_i(\varepsilon_i+s_i\sigma_{Ai})\geq\rho_i(\chi_i)$. If $\varepsilon_i+|\sigma_{Ai}|>\pi$, then $\chi_i=\pi$, and evenness and $2\pi$-periodicity make $\rho_i(\pi)$ the global minimum, so the same inequality holds. Hence, $ r_{Ai}-\rho_i(\varepsilon_i+s_i\sigma_{Ai})
\leq r_{Ai}-\rho_i(\chi_i)$, that proves $\widehat B \leq B$
\end{proof} 
Combining \Cref{lem:dmc_radial} with \Cref{lem:conservative_branch}, enforcing $\widehat B_i\geq0$ is sufficient to guarantee the original DMC constraint in \eqref{eq:dmc_constraint}. Differentiating \eqref{eq:Bi_hat} with respect to time and using \eqref{eqn:rel_dyn_1}-\eqref{eqn:rel_dyn_3}, we obtain,
\begin{align}
\dot{\widehat B}_i =
f_i+g_i\frac{a_A}{v_A},
\label{eq:Bi_hat_dyn}
\end{align}
where $g_i=-s_i\rho_i'\!\left(\varepsilon_i+s_i\sigma_{Ai}\right)$ and $f_i=\dot r_{Ai}-g_i\dot\theta_{Ai}$.

To account for multiple defenders, we aggregate the conservative margins using the smooth minimum as,
\begin{align}
H
\triangleq
-\beta\log\left(
\sum_{i\in\mathcal D}e^{-\widehat B_i/\beta}
\right),
\qquad \beta>0,
\label{eq:H_hat}
\end{align}
that satisfies, $
\min_{i\in\mathcal D}\widehat B_i-\beta\log n
\leq H \leq \min_{i\in\mathcal D}\widehat B_i.$
The parameter $\beta$ controls the smoothness and conservatism of the aggregation. A smaller $\beta$ gives a tighter approximation of $\min_{i\in\mathcal D}\widehat B_i$ but produces a sharper transition between the dominant defenders, whereas a larger $\beta$ provides a smoother distribution of influence among defenders at the expense of the larger approximation gap $\beta\log n$.
Thus, $H\geq0$ guarantees $\widehat B_i\geq0$ for every defender, preserving all the original DMC risk constraints. Differentiating \eqref{eq:H_hat} with respect to time and using \eqref{eq:Bi_hat_dyn} we obtain,
\begin{align}
\dot{H}=F+G\frac{a_A}{v_A},
\label{eq:Hhat_affine}
\end{align}
where $F=\sum_{i\in\mathcal D}w_if_i$, $
G=\sum_{i\in\mathcal D}w_ig_i$, with
$w_i=e^{-\widehat B_i/\beta}/\sum_{j\in\mathcal D}e^{-\widehat B_j/\beta}$. To impose the lateral acceleration constraint, we define the admissible input set as,
\begin{align}
\mathcal U\triangleq[-\bar a,\bar a],
\qquad 0<\bar a<a_{\max}.
\label{eq:direct_input_set}
\end{align}
For a candidate acceleration $a_A\in\mathcal U$, define the risk residual,
\begin{align}
\mathcal R(a_A)=
F+G\frac{a_A}{v_A}+\kappa H,
\qquad \kappa>0.
\label{eq:direct_residual}
\end{align}
Thus, $\mathcal R(a_A)\geq0$ implies
$\dot{H}\geq-\kappa H$. Since
\eqref{eq:direct_residual} is affine in $a_A$, its maximum over $\mathcal U$ is
\begin{align}
\mathcal M(x,t)
&\triangleq\max_{|a_A|\leq\bar a}\mathcal R(a_A)=F+\kappa H
 +\frac{|G|}{v_A}\bar a.
\label{eq:direct_feasibility}
\end{align}
As a result, an admissible acceleration satisfying the risk inequality exists if and only if $\mathcal M(x,t)\geq0$.

We now propose the target interception command as,
\begin{align}
a_A^T
=-\bar a\tanh\!\left(
\frac{(v_AK_I+v_A^2/r_T^c)\sin\sigma_{AT}}{\bar a}
\right),
\label{eq:direct_target}
\end{align}
where $K_I>0$ the controller gain and $r_T^c=\max\{r_{AT},r_c\}$ avoids singular LOS compensation term near the capture radius. The command satisfies $|a_A^T|<\bar a$ and steers the attacker heading toward the attacker to target LOS. 

To avoid violating the risk constraints, we first introduce a bounded preemptive maneuver,
\begin{align}
a_A^p
=\ell a_p\psi\!\left(\frac{|G|}{g_p}\right),
\qquad 0<a_p\leq\bar a,
\label{eq:direct_preemptive}
\end{align}
$g_p\geq g_\ast>0$ denotes the turning-authority threshold that governs the attenuation of the preemptive maneuver, $\ell\in\{-1,1\}$ denotes the escape direction and $\psi(\cdot)$ is a smooth tapering function defined as,
\begin{align}
\psi(z)=
\begin{cases}
1-3z^2+2z^3,&0\leq z<1,\\
0,&z\geq1.
\end{cases}
\label{eq:direct_taper}
\end{align}
Additionally we define, $d>0$ as a user-selected anticipation margin that activates the maneuver when $\mathcal R(a_A^T)<d$. A larger $d$ initiates avoidance farther from the risk boundary, whereas a smaller $d$ retains the target command closer to that boundary. Whenever $\mathcal R(a_A^T)<d$, $\ell\in{-1,1}$ is chosen as the value that yields the larger $\mathcal{R}(a_A^p)$ and is retained for the corresponding avoidance maneuver, until $\mathcal R(a_A^T)\geq d$.

\begin{remark}
For each defender, the sign of $s_i$ is chosen as the value that yields the larger $\widehat B_i$ at $t=0 $ and retained throughout the engagement. In contrast, $\ell$ is selected whenever the anticipation region is entered and retained only for the corresponding avoidance maneuver. It is important to note that $s_i$ only determines the smooth surrogate branch used for risk evaluation and does not impose the physical turning direction, the latter is selected online through $\ell$ whenever the avoidance maneuver is activated.
\end{remark}

We now propose the overall risk-preserving acceleration as,
\begin{align}
a_A^b=
\begin{cases}
a_A^p,
&\mathcal R(a_A^T)<d,\ \mathcal{R}(a_A^p)\geq0,\\[1mm]
a_A^p-\dfrac{v_A \mathcal{R}(a_A^p)}{G},
&\mathcal R(a_A^T)<d,\ \mathcal{R}(a_A^p)<0,
\end{cases}
\label{eq:direct_risk_command}
\end{align}
The first case represents the scenario when the preemptive command already satisfies the risk inequality. Otherwise, the second case provides the correction required to enforce $\mathcal R(a_A^b)=0$. For $G\neq0$, the two branches coincide at $\mathcal \mathcal{R}(a_A^p)=0$, while feasibility ensures that the corrective command remains admissible.

We now propose the lateral acceleration as a convex combination of the target interception and risk-preserving command,
\begin{align}
a_A=\alpha a_A^b+(1-\alpha)a_A^T,
\label{eq:direct_acceleration}
\end{align}
where $\alpha$ denotes a smooth switching function given as,
\begin{align}
\alpha=
\begin{cases}
1,&\mathcal R(a_A^T)\leq0,\\
1-3\xi^2+2\xi^3,
&0<\mathcal R(a_A^T)<d,\\
0,&\mathcal R(a_A^T)\geq d,
\end{cases}
\qquad
\xi\triangleq\frac{\mathcal R(a_A^T)}{d}.
\label{eq:direct_switching}
\end{align}
Thus, the controller transitions smoothly from pure target-interception guidance to pure risk-preserving guidance as $\mathcal R(a_A^T)$ decreases from $d$ to $0$. The cubic interpolation has zero slope at both endpoints, reducing sensitivity to variations in the inactive command near the switching boundaries.

\begin{theorem}
\label{thm:direct_safety}
Consider any relative target-attacker-defender kinematics, if attacker applies lateral acceleration \eqref{eq:direct_acceleration}, such that $H(0)\geq0$ and the pointwise feasibility condition $\mathcal M(x(t),t)\geq0$ holds almost everywhere throughout the engagement. Then the prescribed DMC constraints are preserved, and the lateral acceleration remains within the admissible bounds.
\end{theorem}
\begin{proof}
The first case of
\eqref{eq:direct_risk_command} gives
$\mathcal R(a_A^b)=\mathcal \mathcal{R}(a_A^p)\geq0$ and $|a_A^b|\leq\bar a$ by design of the preemptive maneuver. Consider now the second case of \eqref{eq:direct_risk_command}, for which $\mathcal \mathcal{R}(a_A^p)<0$. If
$G=0$, then $\mathcal R(u)=\mathcal \mathcal{R}(a_A^p)<0$ for every
$u\in\mathcal U$, contradicting $\mathcal M\geq0$. Hence
$G\neq0$, and direct substitution gives
\begin{align}
\mathcal R(a_A^b)
=\mathcal \mathcal{R}(a_A^p)+\frac{G}{v_A}
\left(-\frac{v_A \mathcal \mathcal{R}(a_A^p)}{G}\right)=0.
\label{eq:direct_corrective_residual}
\end{align}
Moreover, $\mathcal M\geq0$ guarantees the existence of some $u\in\mathcal U$ with $\mathcal R(u)\geq0$. Since $\mathcal R$ is affine and $\mathcal \mathcal{R}(a_A^p)<0$, its zero lies between $a_A^p$ and such an admissible input. Therefore, $a_A^b\in\mathcal U$.

If $\mathcal R(a_A^T)\geq d$, then $\alpha=0$, $a_A=a_A^T \in \mathcal{U}$, and
$\mathcal R(a_A)\geq d$. If $\mathcal R(a_A^T)\leq0$, then
$\alpha=1$, $a_A=a_A^b$, and $\mathcal R(a_A)\geq0$. In the transition
region, affinity of \eqref{eq:direct_residual} gives
\begin{align}
\mathcal R(a_A)
=\alpha\mathcal R(a_A^b)
 +(1-\alpha)\mathcal R(a_A^T)\geq0.
\label{eq:direct_blended_residual}
\end{align}
Consequently, \eqref{eq:direct_residual} and
\eqref{eq:direct_blended_residual} yield
$\dot{H}\geq-\kappa H$ almost everywhere. The comparison
principle for absolutely continuous functions gives
$H(t)\geq H(0)e^{-\kappa t}\geq0$. Since $H\leq\min_{i\in\mathcal D}\widehat B_i$, it follows that $\widehat B_i(t)\geq0$ for every defender. By \Cref{lem:conservative_branch}, $B_i(t)\geq\widehat B_i(t)\geq0$, which in turn implies guaranteed preservation of all prescribed DMC constraints using \Cref{lem:dmc_radial}. Since both $a_A^T$ and $a_A^b$ belong to $\mathcal U$ and $\alpha\in[0,1]$, convexity of $\mathcal U$ yields $|a_A|\leq\bar a<a_{\max}.$ This concludes the proof.
\end{proof}

\begin{figure*}[h!]
\centering
\begin{subfigure}[t]{0.245\linewidth}
\centering
\includegraphics[width=\linewidth]{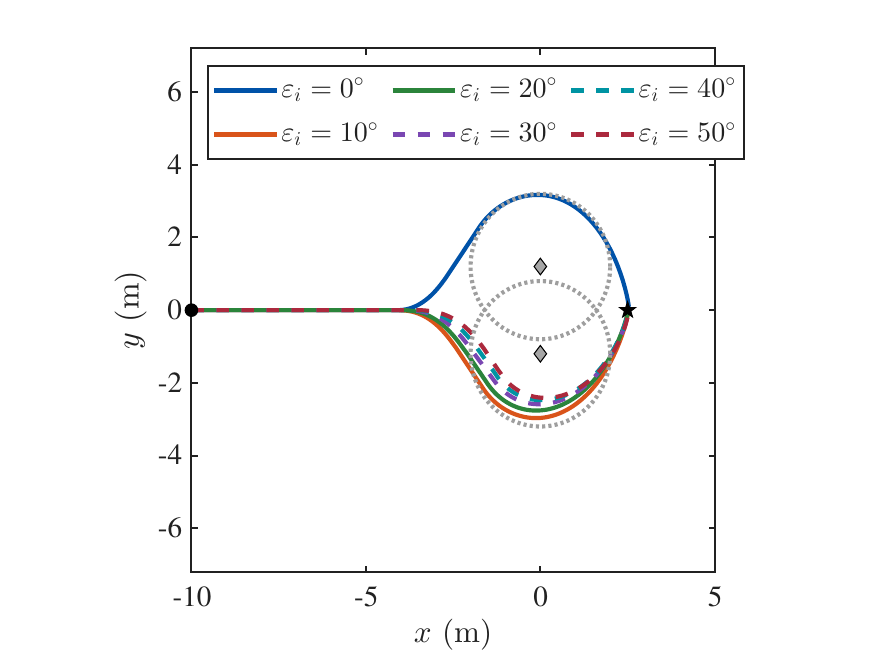}
\caption{Trajectories.}
\label{fig:2a}
\end{subfigure}\hfill
\begin{subfigure}[t]{0.245\linewidth}
\centering
\includegraphics[width=\linewidth]{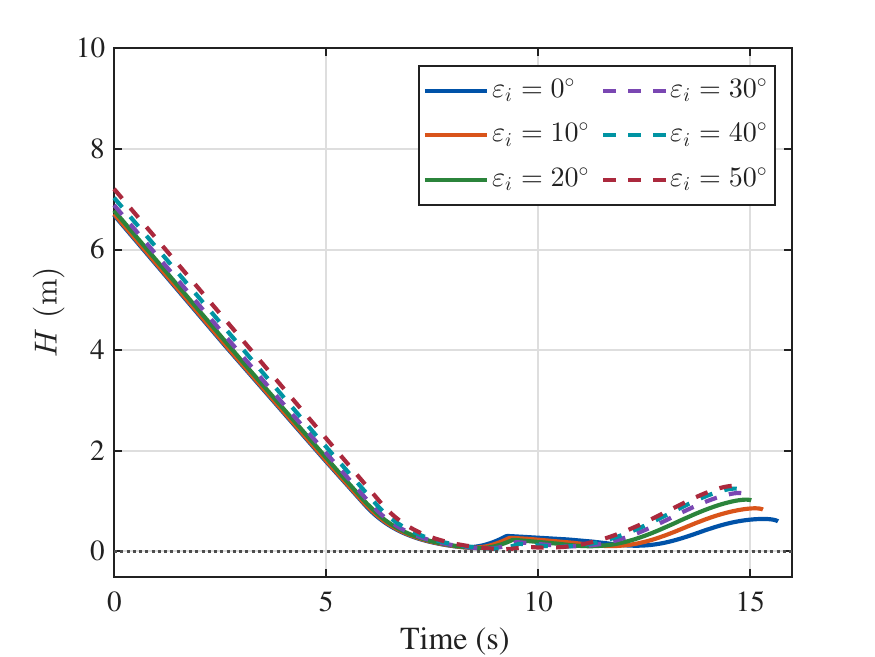}
\caption{Safety constraints.}
\label{fig:2b}
\end{subfigure}\hfill
\begin{subfigure}[t]{0.245\linewidth}
\centering
\includegraphics[width=\linewidth]{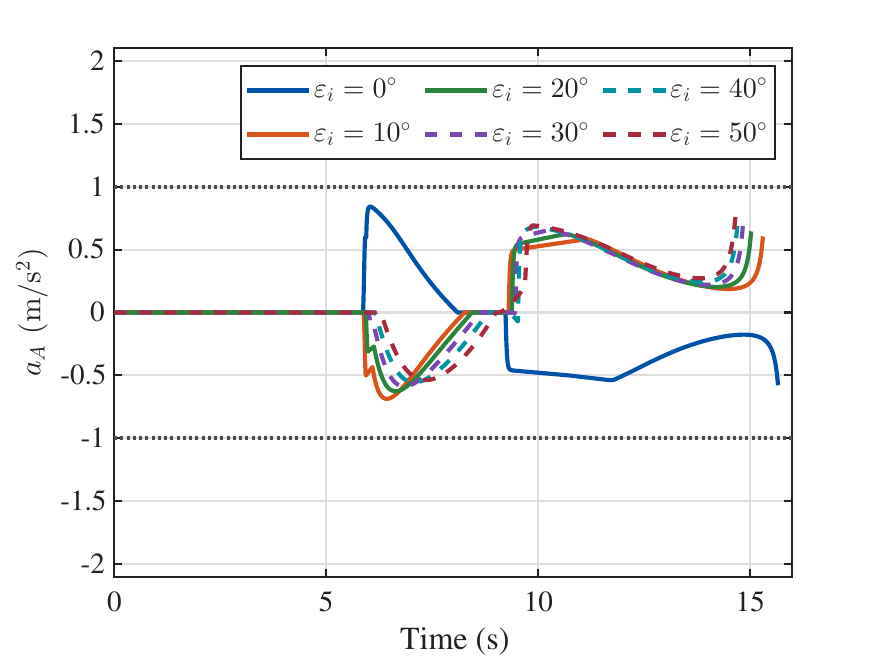}
\caption{Control input.}
\label{fig:2c}
\end{subfigure}\hfill
\begin{subfigure}[t]{0.245\linewidth}
\centering
\includegraphics[width=\linewidth]{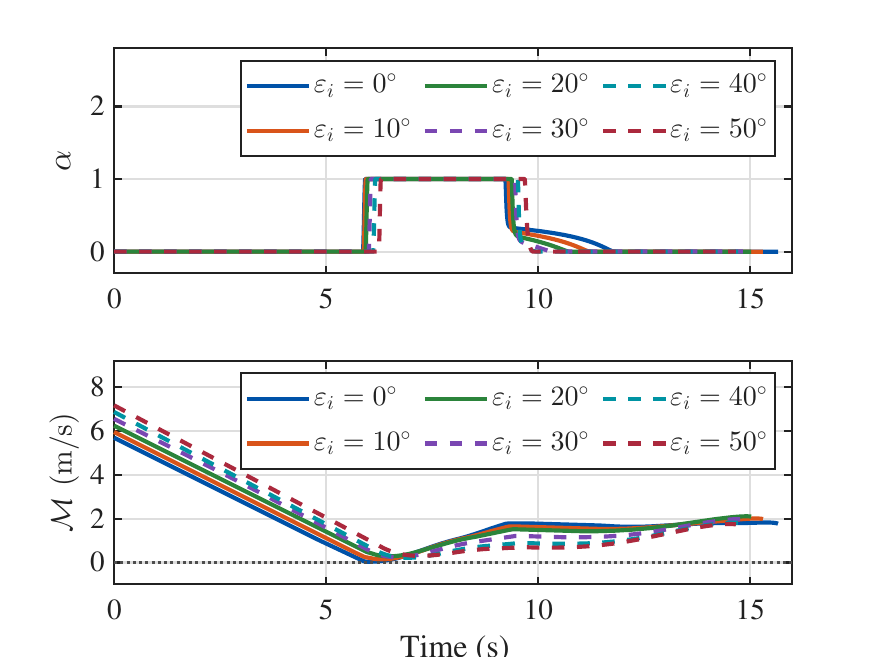}
\caption{Switching function and input feasibility.}
\label{fig:2d}
\end{subfigure}
\caption{Comparison of attacker performance with two stationary defenders.}
\label{fig:stationary_initial_error}
\end{figure*}

\begin{corollary}
\label{cor:direct_capture}
Suppose the conditions of Theorem~\ref{thm:direct_safety} hold and, from some $t_0$ until capture, $\alpha=0$ and
$r_{AT}(t_0)>r_c$. If there exists $\sigma_c\in[0,\pi/2)$ such that $|\sigma_{AT}(t_0)|\le\sigma_c$ and 
\begin{align}
\bar a\tanh\!\left(
\frac{(v_AK_I+v_A^2/r_c)\sin\sigma_c}{\bar a}
\right)
\geq\frac{v_A^2}{r_c}\sin\sigma_c,
\label{eq:direct_capture_condition}
\end{align}
then the attacker reaches $r_{AT}=r_c$ in finite time, with
\begin{align}
t_f-t_0
\leq
\frac{r_{AT}(t_0)-r_c}{v_A\cos\sigma_c}.
\label{eq:direct_capture_time}
\end{align}
\end{corollary}
\begin{proof}
In pure target mode, $a_A=a_A^T$. At the cone boundary
$\sigma_{AT}=\sigma_c$, the stationary-target dynamics give
\begin{align}
\dot\sigma_{AT}
=\frac{a_A^T}{v_A}
 +\frac{v_A}{r_{AT}}\sin\sigma_c.
\label{eq:direct_heading_boundary}
\end{align}
For $s=\sin\sigma_c$ and $q=1/r_{AT}$, the function
$\bar a\tanh((v_AK_I+v_A^2q)s/\bar a)-v_A^2qs$ is nonincreasing in $q$.
Its minimum over $r_{AT}\geq r_c$ therefore occurs at $r_{AT}=r_c$.
Condition \eqref{eq:direct_capture_condition} makes
\eqref{eq:direct_heading_boundary} nonpositive, and symmetry gives the
opposite inward inequality at $\sigma_{AT}=-\sigma_c$. Hence the heading
cone is invariant. Since
$\dot r_{AT}=-v_A\cos\sigma_{AT}\leq-v_A\cos\sigma_c<0$, integration until
$r_{AT}=r_c$ gives \eqref{eq:direct_capture_time}.
\end{proof}

\begin{figure*}[h!]
\centering
\begin{subfigure}[t]{0.245\linewidth}
\centering
\includegraphics[width=\linewidth]{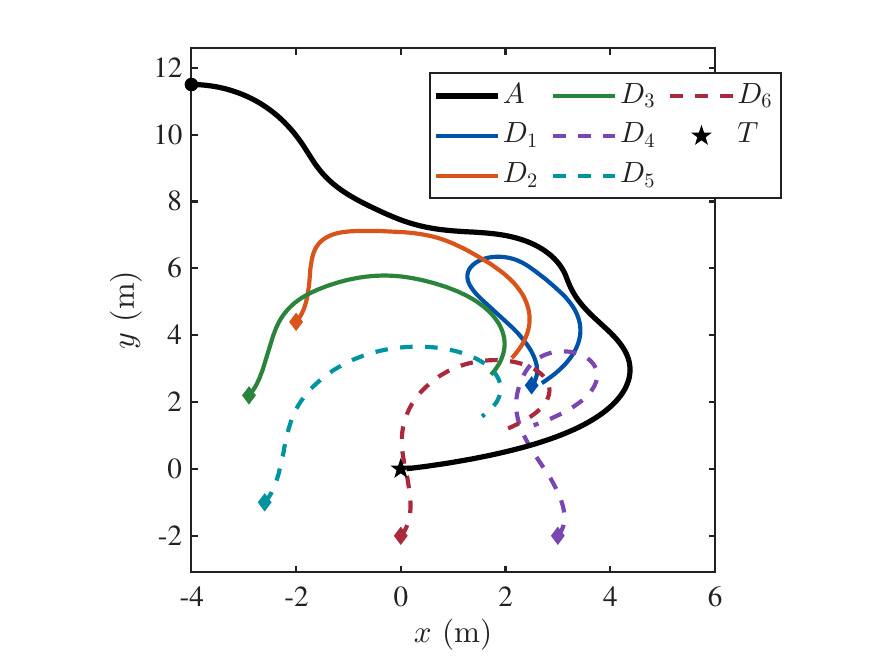}
\caption{Trajectories.}
\label{fig:3a}
\end{subfigure}\hfill
\begin{subfigure}[t]{0.245\linewidth}
\centering
\includegraphics[width=\linewidth]{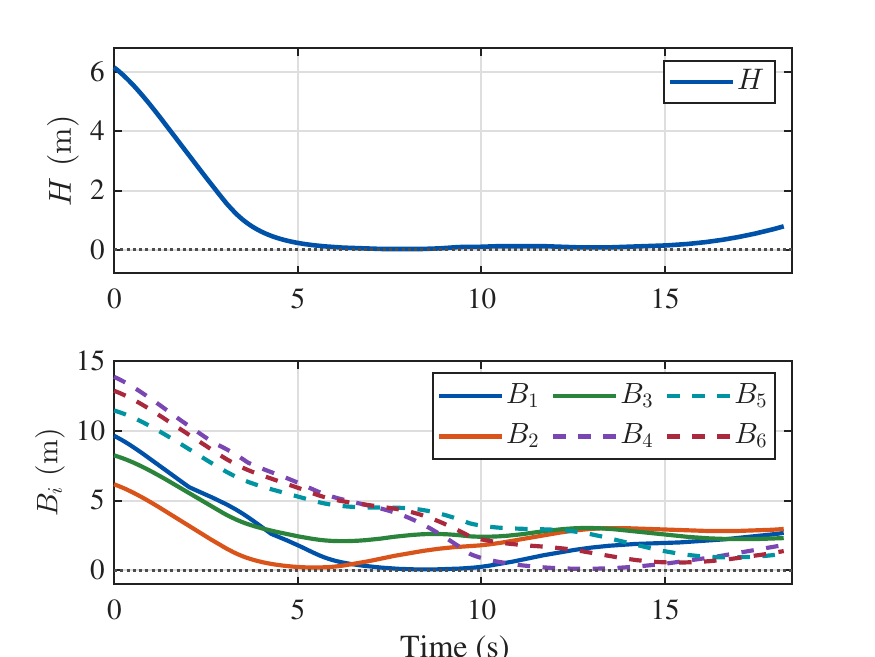}
\caption{Safety constraints.}
\label{fig:3b}
\end{subfigure}\hfill
\begin{subfigure}[t]{0.245\linewidth}
\centering
\includegraphics[width=\linewidth]{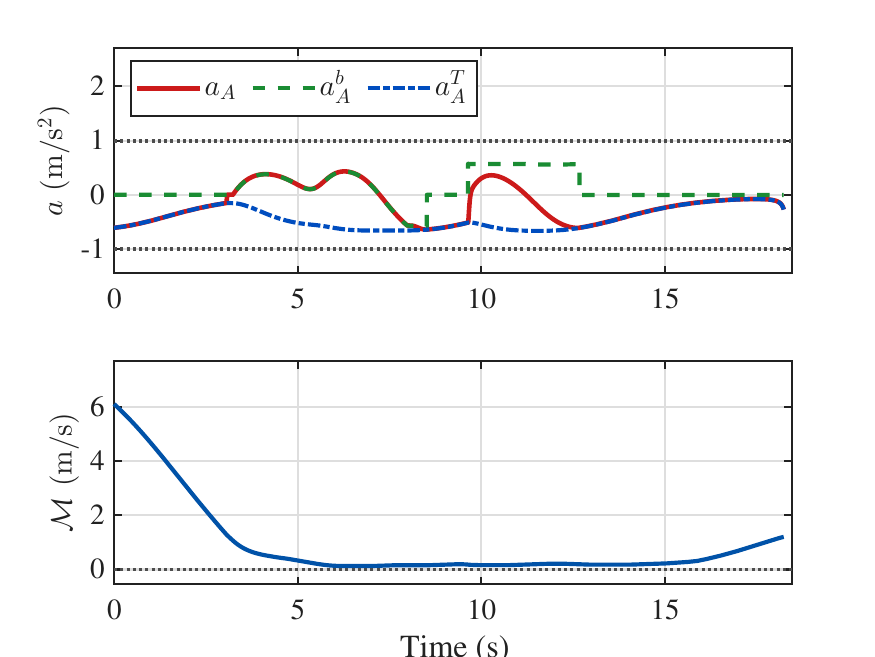}
\caption{Lateral acceleration and input feasibility.}
\label{fig:3c}
\end{subfigure}\hfill
\begin{subfigure}[t]{0.245\linewidth}
\centering
\includegraphics[width=\linewidth]{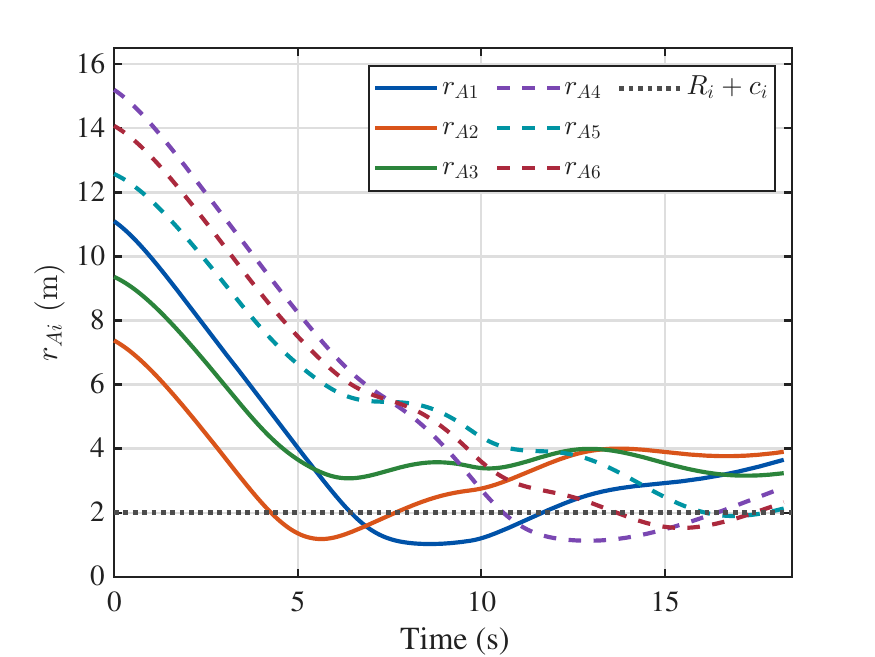}
\caption{Engagement ranges.}
\label{fig:3d}
\end{subfigure}
\caption{Comparison of attacker performance with six moving defenders and prescribed tolerance $\varepsilon_i=50^\circ$.}
\label{fig:moving6}
\end{figure*}

\section{Simulations}
\label{sec:sim}
We now evaluate the proposed guidance strategy for stationary and moving-defender engagements. In all simulations, the attacker speed is $v_A=1\,\mathrm{m/s}$. The defenders move at $v_i=0.5\,\mathrm{m/s}$ in the moving cases and remain stationary otherwise. Their parameters are $R_i=1.5\,\mathrm{m}$, $c_i=0.5\,\mathrm{m}$, and $\mu_i=0.7$. The guidance parameters are $\bar a=1\,\mathrm{m/s^2}$, $a_{\max}=1\,\mathrm{m/s^2}$, $r_c=0.1\,\mathrm{m}$, $\beta=0.5$, $\kappa=1$, $K_I=0.7$, $d=0.05$, $a_p=0.6\,\mathrm{m/s^2}$, and $g_p=0.8$.

For the first set of results, We consider two stationary defenders for various DMC tolerances $\varepsilon_i=0^\circ,10^\circ,20^\circ,30^\circ,40^\circ,$ and $50^\circ$, with dotted circles in the trajectory plots depicting the maximum engagement range $R_i+c_i=2\,\mathrm{m}$. The defenders are located at
$[0,\pm1.2]^\top\,\mathrm{m}$, with the attacker initialized at
$[-10,0]^\top\,\mathrm{m}$ and the target at
$[2.5,0]^\top\,\mathrm{m}$. Figure~\ref{fig:stationary_initial_error} compares the resulting trajectories for
$\varepsilon_i=0^\circ,10^\circ,20^\circ,30^\circ,40^\circ,$ and $50^\circ$. Fig.~\ref{fig:2a} illustrates the attacker trajectories demonstrating progressively greater penetration of the defenders' engagement ranges as $\varepsilon_i$ increases exhibiting the reduction in conservatism of the avoidance maneuver. The safety measure $\widehat H$ is depicted in Fig.~\ref{fig:2b} remains positive for all cases, illustrating preservation of the prescribed DMC constraints. The actual lateral acceleration also remains within the admissible limit as shown in Fig.~\ref{fig:2c}. Figure~\ref{fig:2d} shows the switching function activating the avoidance maneuver after approximately $5\,\mathrm{s}$ to preserve the prescribed risk constraints, followed by sustained activation of the target-interception maneuver toward the end of the engagement for successful target capture. The second subplot shows that the feasibility variable remains positive throughout the engagement to demonstrate that the bounded-input feasibility condition is maintained at all times.
\begin{figure}[t]
\centering
\begin{subfigure}[t]{0.48\columnwidth}
\centering
\includegraphics[width=\linewidth]{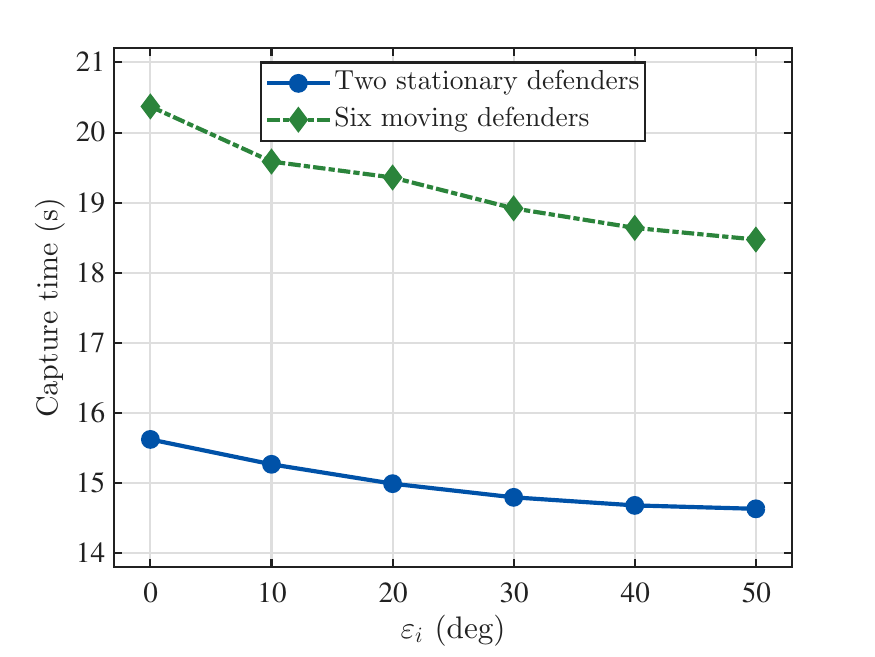}
\caption{Capture time.}
\label{fig:4a}
\end{subfigure}\hfill
\begin{subfigure}[t]{0.48\columnwidth}
\centering
\includegraphics[width=\linewidth]{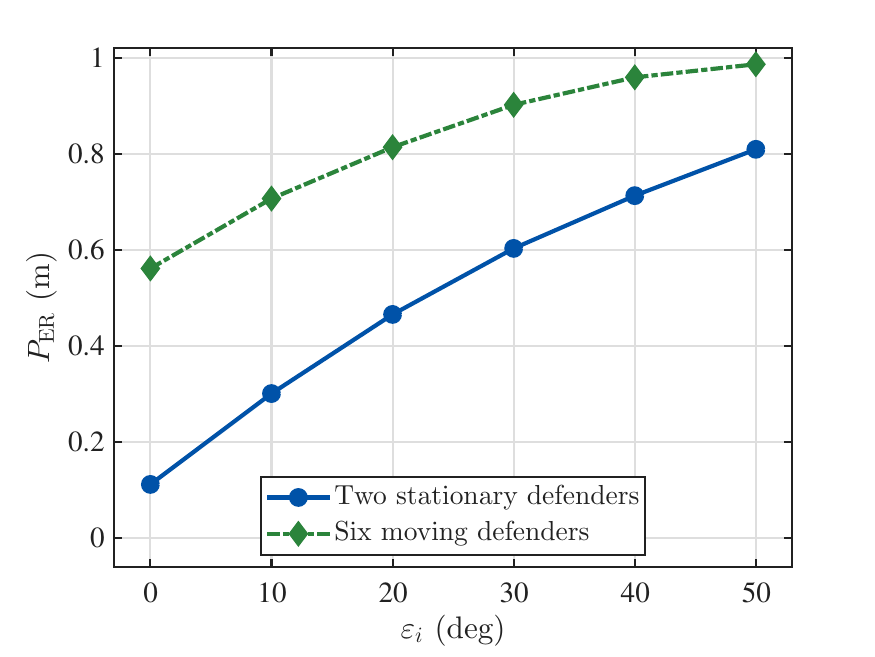}
\caption{Maximum engagement-range penetration.}
\label{fig:4b}
\end{subfigure}
\caption{Effect of the prescribed risk tolerance on interception performance.}
\label{fig:performance_sweep}
\end{figure}

Fig.~\ref{fig:moving6} illustrates the results for six defenders moving at $0.5 \mathrm{m/s}$ with $\varepsilon_i=50^\circ$, where defenders adopt an aggressive stance and utilize pure pursuit to intercept the attacker. Fig.~\ref{fig:3a} illustrates the attacker trajectory successfully avoiding the defenders and intercepting the target. The aggregate and individual safety measures remain always positive throughout the engagements, as observed in Fig.~\ref{fig:3b}, demonstrating that the aggregation preserves the individual safety constraints.  Fig.~\ref{fig:3c} depicts the attacker's acceleration remaining within the permissible bounds and switching between the avoidance maneuver and target interception maneuver, while second subplot demonstrates the input-feasibility. The first subplot demonstrates the smooth switching between target interception and threat avoidance command, while the second subplot illustrates the commanded acceleration that remains bounded throughout. Fig.~\ref{fig:3d} illustrates the attacker's ability to penetrate the defenders' maximum engagement range, as the attacker-defender distances fall below the $2,\mathrm{m}$ threshold indicated by the dotted black line. Finally, the Fig.~\ref{fig:performance_sweep} illustrates risk-performance trade-off for stationary and moving defenders. As shown in Fig.~\ref{fig:4a}, the target intercept time decreases with increasing $\varepsilon_i$, indicating that larger DMC tolerances allow shorter target interception trajectories.  Fig.~\ref{fig:4b} compares the engagement-range penetration calculated as $P_{\mathrm{ER}}\triangleq
 \max_{i,t}[R_i+c_i-r_{Ai}(t)]_+$ that increase with $\varepsilon_i$. Therefore, a larger DMC tolerance permits deeper penetration of the engagement range and reduces the avoidance detour, resulting in improved interception performance.

\section{Conclusions}
This work addressed safe target interception in contested environments by allowing the attacker to trade strict defender avoidance for a prescribed level of recoverable geometric exposure. The proposed formulation uses the DMC-based risk tolerance to distinguish between unsafe penetrations that require excessive recovery maneuvers and admissible penetrations from which the attacker can still recover. The developed guidance law combines multiple defender constraints, anticipates loss of effective turning authority, while respecting the available lateral-acceleration limits. The closed-loop results show that the prescribed risk bounds and actuator constraints can be maintained while preserving bounded guidance behavior and enabling target interception under appropriate engagement conditions. Future work will consider adaptive defender-dependent risk tolerances, uncertainty in defender parameters and measurements, and three-dimensional experimental validation.

\bibliographystyle{ieeetr}
\bibliography{references}

\end{document}